\documentclass[11pt,letterpaper]{article}

\usepackage[margin=1in]{geometry}
\usepackage{hyperref}

\usepackage{hyperref,enumitem}
\hypersetup{
	colorlinks=true,
	citecolor = blue       %
}
\setlist{noitemsep,leftmargin=\parindent,topsep=2pt}
\setlist{noitemsep,topsep=2pt}
\usepackage{natbib}
\usepackage{url}            %
\usepackage{xspace}
\usepackage{graphicx}
\usepackage{xcolor}         %
\usepackage{amsmath,amssymb,amsthm}
\usepackage{thm-restate}
\usepackage{booktabs} %
\usepackage[font=small,labelfont=bf]{caption}
\usepackage{subcaption}
\usepackage{multirow}
\usepackage{authblk}

\usepackage{comment}
\usepackage{algorithm}
\usepackage{wrapfig}
\usepackage{algorithmic}
\newtheorem{theorem}{Theorem}[section]

\newtheorem{assumption}[theorem]{Assumption}

\newtheorem{proposition}[theorem]{Proposition}

\newtheorem{definition}[theorem]{Definition}

\def\min{\qopname\relax n{min}}
\def\max2{\qopname\relax n{max2}}
\def\max{\qopname\relax n{max}}

\def\argmax{\qopname\relax n{argmax}}

\def\Ex{\qopname\relax n{\mathbb{E}}}

\newcommand{\expect}[2]{\Ex_{#1}\left[#2\right]}

\def\1{\mathbb{I}}

\def\B{\mathcal{B}}

\def\F{\mathcal{F}}

\def\M{\mathcal{M}}

\def\S{\mathcal{S}}

\def\Z{\mathcal{Z}}
\def\RR{\mathbb{R}}

\def\part{P}

\def\prom{\textit{Prom} }
\def\Prom{\textsc{Prom} }
\def\llm{\textsc{llm}}
\def\gen{\text{Gen}}
\def\ROUGE{\textsc{ROUGE}}

\newenvironment{lp*}{\begin{equation*}  \begin{array}{lll}}{\end{array}\end{equation*}}

\newcommand{\pctr}{\text{pctr}}
\newcommand{\posnorm}{\text{pos}\_\text{norm}}
\newcommand{\ppctr}{\text{pctr}^{\text{pos}}}
\newcommand{\fpctr}{\text{pctr}^{\text{final}}}
\newcommand{\vbid}{\vec{b}}
\newcommand{\bid}{\text{b}}
\newcommand{\vpctr}{\overrightarrow{\pctr}}
\newcommand{\vecpm}{\overrightarrow{\text{ecpm}}}
\newcommand{\ecpm}{\text{ecpm}}

\newcommand{\vposnorm}{\overrightarrow{\posnorm}}
\renewcommand{\cite}{\citep}

\newcommand{\ctrllm}{\text{ctr}^{\llm}}

\begin{document}

\title{Auctions with LLM Summaries}

\author{Kumar Avinava Dubey\footnote{Author list follows alphabetical order.}}
\author{Zhe Feng$^*$}
\author{Rahul Kidambi$^*$}
\author{Aranyak Mehta$^*$}
\author{Di Wang$^*$}

\affil{Google Research, Mountain View \authorcr \texttt{avinavadubey,zhef,rahulkidambi,aranyak,wadi@google.com}}

\date{April 11, 2024}

\maketitle
\begin{abstract}
We study an auction setting in which bidders bid for placement of their content within a summary generated by a large language model (LLM), e.g., an ad auction in which the display is a summary paragraph of multiple ads. This generalizes the classic ad settings such as position auctions to an LLM generated setting, which allows us to handle general display formats. We propose a novel factorized framework in which an auction module and an LLM module work together via a prediction model to provide welfare maximizing summary outputs in an incentive compatible manner. We provide a theoretical analysis of this framework and synthetic experiments to demonstrate the feasibility and validity of the system together with welfare comparisons.
\end{abstract}

\section{Introduction}\label{sec:intro}
The advent of large language model (LLM) technology has the potential to change the user experience of online services such as internet search, online recommendations~\cite{geng2022recommendation}, or shopping~\cite{fan2023recommender}. For example, search platforms and apps, e.g., Microsoft Bing~\cite{AIB} and Google Search~\cite{SGE}, %
have already experimented with generative AI tools to provide augmented search summarization to facilitate users' search experience. Such summarization (e.g., based on retrieval augmented generation RAG~\cite{RAG2020}) can sometimes provide an efficient way for users to gain useful information in a more condensed space.

\begin{figure}[t]
\centering
\includegraphics[width=0.49\textwidth]{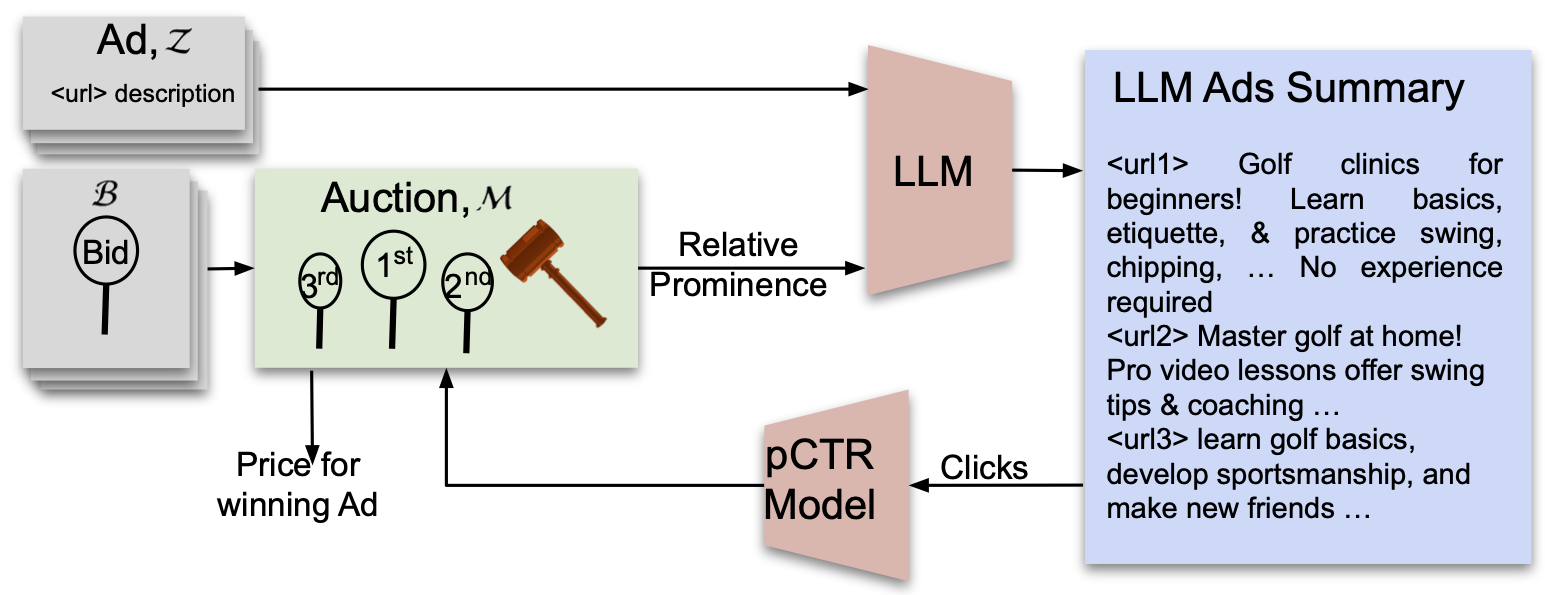}
\caption{\small Factorized model for Auctions with LLM Summaries.}
\label{fig:factor_model}\vspace{-3mm}
\end{figure}

For queries of a commercial nature, search platforms respond with relevant online advertising. Online search advertising has provided a means not only to connect buyers and sellers, but also to support free internet services to users. Given the exciting potential of LLMs to summarize multiple sources of content and provide a succinct and informative output, it is natural to ask how LLM technology can help improve online advertising.

In the ever-evolving landscape of online advertising, auction design has been a critical component towards improving the effectiveness and efficiency of ad delivery. A well-designed auction mechanism not only provides revenue for the platform but also ensures relevancy and value for users and advertisers alike. It efficiently allocates ad impressions to the right audiences, creating high value to advertisers and fostering a positive user experience. While auctions have proved to be a flexible method in various fixed settings such as a single slot~\cite{vickrey1961counterspeculation}, position auctions~\cite{Varian06,EOS07} or list of rich ads~\cite{aggarwal2022simple}, auctions have not been studied in general summarization settings. 

Here, we consider a general setting in which an LLM takes as input a set of ads (and ad assets such as creatives and web pages) and returns a summarized paragraph which can be more helpful to the user, e.g., to compare and contrast product features, use cases, or price. This immediately raises a few challenges compared to the currently used static settings. For example, in a position auction the auction directly determines the position of each ad, and furthermore there is a fixed text creative for each ad. Since the position based predicted click-through rates (CTRs) are known to the auction at decision time, the auction knows an estimate of the expected clicks, welfare, revenue for each possible allocation (permutation of ads), and hence can make efficient decisions on allocation and pricing. In the LLM-based setting that we introduce, the final summarized ad text and relative user attention to the individual ads are all determined at run time by the LLM. Thus the LLM black-box lies between the auction's decision and what the user sees, so the auction can not fully control the latter. How can an auction then choose allocations and prices to maximize welfare or revenue, while also providing good auction incentives in such a setting? That is the question we study in this paper.

Our goal is to design an auction framework which can accommodate content summarization via LLMs in quite general summarization settings. We want a system in which the auction and the LLM work hand-in-hand to provide good properties: Firstly, the LLM summaries should be succinct and accurate, and have good entailment to the ad assets. Secondly, the system should provide good incentive properties to the agents (advertisers). For example, a higher bid should result in higher user-attention and higher click-through rate for an ad. Formally, we want the entire mechanism to be Incentive Compatible (IC). Thirdly, we want the mechanism to provide efficient outcomes with high social welfare, leading to high value to advertisers and users.

While we use ad auctions as the guiding application, we note that our formulation and results apply more generally to any summarization task where the individual content items are owned by agents, and where (a) each agent derives a benefit to have more prominence in the summary, (b) the platform's goal is to maximize some welfare objective which incorporates the quality of the content items and the agent's values expressed via bids, and (c) the quality or relevance of the content items for user queries is learned through a prediction model. For example, one may envision a {\bf recommendation setting} in which individual content providers bid to be shown with higher prominence. In a different setting, the content items could be represented by platform-internal systems representing different objective functions. In such a setting, the pricing component of the auction is not needed. An example of such a setting is aggregation of user reviews of a product or service in accordance with some notion of user reputation or quality score.

\subsection{Our Contribution} 
\label{sec:intro-contributions}
\paragraph{Factorized Model for Auctions with LLM Summaries.} 
We introduce a general problem of running auctions to generate summaries of $k$ ads, with an LLM in the loop. By a summary of ads we mean any collective representation of ads, ranging from an ordered list of fixed ad assets (as in a position auction) to a combined summary of all ads with potential images or videos.

We then develop a framework for an auction to work for such general summary requirements. Our framework contains two main modules, an auction module and an LLM module. The auction module takes as input the bids, and ad qualities, as well as click predictions from a predicted click-through rate (pCTR) module, and outputs a \emph{prominence allocation} and prices. The former is a key abstraction that we introduce to be the interface between the auction and LLM module. The LLM module takes as input the prominence allocation from the auction module, and generates a summary (in the required format). We can take the prominences to be real numbers, although the general definition allows for abstract prominence spaces.

We then describe three sufficient properties for such a factorized model to work effectively: 
\begin{enumerate}[leftmargin=*]
    \item[(a)] The auction's allocation function should be monotonic in bids, i.e., a higher bid results in a higher prominence number.

\item[(b)] The LLM should have a \emph{faithfulness} property which requires the LLM output follows the auction's instructions via the prominences. One can think of the LLM as computing a function which converts the prominences to real user attention. Faithfulness requires that this function is monotonic in each ad's prominence, i.e., as an ad's prominence allocation increases it receives higher user attention. It also requires the user-consideration achieved by the different ads to be proportional to their prominences.%

\item[(c)] The CTR prediction module works with the abstract prominences as features and learns the above function implemented by the LLM module. Formally, we require that the CTR module generate an unbiased estimation of the click-through rate function given the prominence, under expectation of the LLM generation. 
\end{enumerate}

We show that these three properties allow us to close the loop between the three modules, and the auction can be oblivious to the function implemented by the LLM.

Utilizing advanced prompt design techniques \citep{zhou2022least,madaan2022text}, we strategically craft prompts that align the the LLM with the auction's instructions. We provide synthetic experiments to verify the feasibility and validity of our factorized model approach for a specific simple choice of summarization, which we call \emph{dynamic word length}. Experiments also demonstrate the efficiency of our mechanism in that it can produce more efficient outcomes compared to a static position auction or even a fixed-length auction that we define later.

Also for the specific example of dynamic word length summarization, we show that for a simple yet very realistic family of pCTR models, we can fully characterize the (exact) welfare-maximizing auction in that setting. This is in contrast to the case of rich format position auction where finding exact welfare-maximizing auction is typically challenging due to the combinatorial nature of the feasible space. Intuitively, LLM empowers us to expand the feasible space to a much larger continuous space. This both increases the optimal welfare and makes the optimization task easier, which at a high level is analogous to relaxing integer program to linear program.

\subsection{Related Work}
There has been rich literature on position auctions in online advertising, e.g.,~\cite{Varian06, EOS07}, including rich-ad formats~\cite{aggarwal2022simple,cavallo-richads}, and our work can be regarded as an extension of current position auctions to incorporate new formats based on generative AI. 

Mechanism design with Large Language Models (LLM) is a very new but rising field. A 
work by \citet{duetting2023mechanism} propose a token-based auction framework for LLM agents, where the bidders use LLM to generate ads and the auction applies distributional aggregation across different LLMs from the bidders to generate the ads paragraph in a token-by-token. In~\cite{duetting2023mechanism}, the allocation and payment rules are both operated in a token-by-token manner, where the tokens are randomly generated by an aggregation of each agent's LLM. Whereas, in our paper, we propose a factorized framework that contains an auction module and an LLM module, in which the allocation and payment are still decided by auction module and the LLM module will only be used to generate summaries following the guidance from the allocated \textbf{prominence} by the auction module. In another work, \citet{feizi2023online} discuss general challenges and opportunities of online advertising in the age of generative AI.

We will be using prompting to ensure LLMs stick to auction's instructions and to generate the synthetic dataset containing queries and ads. Prompting is a fast and efficient method for downstream application of LLMs and has been extensively studied \citep{nye2021show, wei2022chain, zhou2022least, wang2022self, wang2023plan, madaan2023self}. We use Chain-of-Thought (CoT) \citep{zhou2022least}, and provide few shot reasoning as intermediate steps to improve performance to downstream goal.  Multiple variants of CoT including Zero-shot CoT \citep{kojima2022large}, self consistency \citep{wang2022self}, Tree of Thoughts \citep{yao2023tree}, Graph of Thoughts \citep{besta2023graph} further extend the reasoning capability of CoT methods. Among them, we use iterative reasoning to ensure LLMs stick to auction instruction.

\section{Model and Preliminaries}\label{sec:model}

As mentioned in Section~\ref{sec:intro-contributions}, 
we propose a factorized model consisting of an auction module, an LLM module, and a click-through rate prediction module. A high-level schematic of this factorized model is provided in Figure~\ref{fig:factor_model}. We next describe the input-output characteristics of the three modules in more detail.

\subsection{Auction Module}
We consider $n$ different bidders competing for the ad summaries shown by LLM. Each bidder $i\in [n]$ has a \emph{private} value $v_i \in \mathbb{R}_{\geq 0}$ when its ad gets clicked. Follow the standard Bayesian mechanism design literature, we assume $v_i \sim \F_i$ and that the valuation distribution $\F_i$ is known to the other bidders and to the auctioneer. Denote $\F = \times_i \F_i$ be the joint distribution of valuation profile $v = (v_1, v_2,\cdots, v_n)$.  Let $b_i \in \B_i \subseteq \mathbb{R}_{\geq 0}$ denote bidder $i$'s bid which can be different from the true value $v_i$. Without loss of generality, we assume value and bids are from the same space, i.e., $v_i\in \B_i$. 

Denote $b = (b_1, \cdots, b_n) \in \B = \times_i \B_i$ as the bid profile of $n$ bidders. Following the standard auction design literature, we define $b_{-i}$ as the bid profile of the other bidders except for bidder $i$. For each bidder $i$, let $z_i \in \Z_i$, be a feature which contains assets from the bidder, such as the ad creative and landing page. Let $z = (z_1, \cdots, z_n)\in \Z = \times_i \Z_i$ be the feature profile of all the bidders' ads, which will be used as context in the LLM. 

{\bf Relative Prominence:} As usual, the auction mechanism $\M = (x, p)$ contains an allocation rule $x$ and payment rule $p$, both being functions of the bids and pCTRs. Different from standard auction design, the allocation rule $x$ in our prominence-based auction module does not directly specify the allocation but outputs the \emph{relative prominence} of each bidder's ad. This represents the relative importance the ad is supposed to get in the LLM generated summary, and will be taken as input to the LLM module.  

Specifically, we define $\Prom$ as an abstract space of prominences which is the interface language between the auction and the LLM modules -- it is the range space of the allocation function $x$ and the input space for the LLM. The auction also needs to know an estimate of the click-through rate that each ad would get if it output a particular $\prom \in \Prom$. This is enabled by the pCTR module (Sec.~\ref{sec:model-pctr-module}), which gives the auction a map $\pctr: \Prom \times \Z \rightarrow [0,1]^n$.
The {\bf allocation function} of the auction is now a function $x: \B \times \pctr \rightarrow \Prom$ specifying the allocation of prominence. 
We assume $\Prom$ has a well-defined tuple of order relations $\succeq := (\succeq^1, \cdots, \succeq^n)$, in which $\succeq^i, \forall i \in [n]$ specifies the preference of each bidder $i$ over all possible relative prominences. %
The {\bf payment rule} $p = (p_1, \cdots, p_n)$ specifies the (expected) payment  for all bidders given submitted bids, where $p_i: \B \times \pctr \rightarrow \RR_{\geq 0}$ maps the bid profile to a non-negative payment.

Relative prominence is a general and \emph{abstract} concept, as long as it has a well-defined order relation $\succeq$ and the LLM can easily follow its guidelines. For example, the prominence can represent for each ad, a tuple of the space and the attractiveness of its creative. In such a setting, the LLM could first follow the requirements of the space and then generate the appropriately attractive creatives using different multimodal formats for each ad. To simplify the presentation and to formalize a concrete setting, we focus on a simple structure of relative prominence throughout the rest of the paper:%
\begin{definition}The Relative Prominence $\Prom$ is the set of points $\prom = (\prom_1, \cdots, \prom_n)$, with $\sum_{i\in [n]} \prom_i  \leq 1$. Here, $\prom_i$ represents bidder $i$'s allocated prominence.%
\end{definition}

\subsection{LLM Module}\label{sec:model-llm}
The LLM module can be abstracted as a function to map the allocated prominence of all ads to a summary. 
Formally, the LLM is a function $\gen: (\Prom, \Z) \rightarrow \S$, where $\S$ represents the set of possible combined summaries of the $n$ bidders.

For the factorized system described in Figure~\ref{fig:factor_model} to work efficiently, we would like the LLM module to satisfy the following property.

\begin{definition}[Faithfulness]
Given a set of input ad features $z\in \Z$ and a set of input prominences $\prom \in \Prom$, the LLM output summary $s = \gen(z, \prom)$ should be such that when a user reads $s$, then ad $i$ gets an amount of consideration proportional to $\prom_i$.
\end{definition}

This definition is informal due to the absence of a mathematical formulation of user consideration. However, the Faithfulness property can be evaluated through testing on a panel of users or paid evaluators. 
Informally, the property asks the LLM to implement relative prominence instructions given by the auction.
Note that the faithfulness property implies \emph{(strict) monotonicity}: as an ad's prominence increases, the user attention it gets also increases. We will show in Sec.~\ref{sec:properties} that monotonicity is a sufficient condition to make the system incentive compatible.

\subsection{pCTR Module}
\label{sec:model-pctr-module}
As in classic ad auctions, we need a predicted click-through rate (pCTR) module to allow the auction to make efficient allocation decisions\footnote{Click-through prediction is required to enable per-click payment schemes. We note that we can similarly incorporate other predictions such as conversion rates in our framework.}. 
The CTR of each ad is not just a function of its intrinsic quality but is also modulated by its representation in the ad summary generated by the LLM. This creates a difficulty in designing the auction's allocation function since the generated summary is not known at auction time and the LLM is not a deterministic function. We get around this in our factorized model by requiring that the pCTR model only predict the average CTR for each ad given the relative prominence $\prom$. Formally, the pCTR model $\pctr: \Prom \times \Z \rightarrow [0, 1]^n$ maps the ads feature profile $z\in \Z$ and the allocated prominence $\prom \in \Prom$ to an \emph{average} click-through rate that takes the expectation over the randomness of the LLM generation. Let $\pctr_i: \Prom \times \Z \rightarrow [0, 1]$ be the average pCTR model of bidder $i$.

Given the ad summaries generated by LLM, let $\ctrllm: \S \rightarrow [0, 1]^n$ model the expected CTR for $n$ ads. Similarly, denote $\ctrllm_i: \S \rightarrow [0, 1]$ as the expected CTR function of bidder $i$. %
\begin{assumption}[Unbiased Estimation of pCTR Model]\label{ass:unbiased-pctr}
$\forall i\in [n], \prom = (\prom_1, \cdots, \prom_n) \in \Prom, z \in \Z$, we have,
\begin{eqnarray*}
\pctr_i(\prom, z) = \expect{s \in \S: s\sim \gen(\prom, z)}{\ctrllm_i(s)}
\end{eqnarray*}
\end{assumption}
Note that this unbiased estimation of pCTR model in auction stage is important for our factorized model to work as intended, given that the auction relies on this prediction model to make its allocation and pricing decisions. %

Given the above property, we are now ready to define the incentive compatibility property of the factorized model end-to-end, i.e., truthfully reporting values are the optimal bidding strategies for all bidders:
\begin{definition}[Incentive Compatibility of the Factorized Model]
\label{def:e2e-ic}
A factorized model with a prominence-based auction mechanism $\M = (x, p)$ and a LLM $\gen$ is incentive compatible if and only if, $\forall z\in \Z, \forall b\in \B, \forall i\in [n]$ and $b'_i\neq b_i$, 
\begin{align*}
\expect{s\sim \gen(x(b), z)}{\ctrllm_i(s) \cdot b_i - p_i(b)} \geq \expect{s'\sim \gen(x((b'_i, b_{-i})), z)}{\ctrllm_i(s')\cdot b_i - p_i((b'_i, b_{-i}))},
\end{align*}
which, using Assumption~\ref{ass:unbiased-pctr}, is equivalent to,
\begin{align*}
\pctr_i(x(b))\cdot b_i - p_i(b) \geq \pctr_i(x((b'_i, b_{-i})))\cdot  b_i - p_i((b'_i, b_{-i})).
\end{align*}
\end{definition}
Note that for simplicity on notation, we drop the dependence of the allocation and payment functions $x$ and $p$ on the $\pctr$ map.

\subsection{The model is a generalization of Position Auctions}
The model described in this section strictly generalizes the current industry standard of Position Auctions~\cite{Varian06,EOS07}. This can be seen by appropriately choosing the different model components and variables as defined above to fit the setting of a simple position auction (without any LLM in the loop). This is presented in Table~\ref{table:pos-auc} below. Besides showing that the Position Auction is a (very) special case of our general model, the table may also help provide an intuitive understanding of our abstract model.

\begin{table}[h]
\centering
\begin{small}
\begin{tabular}{l  l} 
\toprule
 General Model  (this paper) & Instantiation for classic Position Auctions\\
 [0.5ex] 
 \midrule
Summary ``LLM" & An ordered list $s$ of fixed ad creatives. Sort $\{\prom_i\}_{i\in [n]}$ \\
& and display the ads in that order. \\ 
$\ctrllm_i(s)$ &  Average clicks for ad $i$ in its position in the list $s$.\\
$\pctr_i(\prom, z)$ & Predicted clicks for ad $i$ in the permutation $\prom$.\\ [1ex] 
 \bottomrule
\end{tabular}
\end{small}
\caption{Model instantiation for classic Position Auctions.}
\label{table:pos-auc}
\end{table}
In Section~\ref{sec:case-study}, we consider a more general instantiation of the general model, which we call Dynamic Word Length Summary (DWLS). In that case each of the above components are more general that the position auction instantiation.

\section{Prominence-based Auction Design}\label{sec:mechanism}
\label{sec:properties}

In this section, we consider the design problems in the factorized prominence-based auctions. First, we provide conditions under which we get incentive compatibility.
Second, we show the factorized model and above sufficient condition is indeed without loss generality, by proving a ``revelation principle" type result. Finally, we discuss the 
welfare-maximizing auction design in general setting. 

\subsection{Incentive Compatibility}

\begin{definition}[monotone Allocation]\label{def:monotone-alloc}
An prominence-based auction mechanism $\M = (x, p)$ has a monotone allocation function if $\forall i \in [n], b_{-i}$ and $b_i \geq b'_i$, $x((b_i, b_{-i})) \succeq^i x((b'_i, b_{-i}))$.
\end{definition}

\begin{definition}[monotonic LLM]\label{def:monotone-llm}
An LLM $\gen$ is said to be monotonic if $\forall i\in [n], z\in \Z, \prom, \prom'$, and $\prom \succeq^i \prom'$,
\begin{eqnarray}\label{eq:monotone-llm}
\expect{s \in \S: s\sim \gen(\prom, z)}{\ctrllm_i(s)} \geq \expect{s \in \S: s\sim \gen(\prom', z)}{\ctrllm_i(s)}
\end{eqnarray}
\end{definition}

We next characterize the incentive compatibility of factorized model.
\begin{proposition}[Incentive Compatibility of Prominence-based Auctions]\label{prop:ic-prom-auction}
Given a monotonic LLM  (Def.~\ref{def:monotone-llm}) and an unbiased $\pctr$ module (Def.~\ref{ass:unbiased-pctr}), a prominence-based auction $\M = (x, p)$ is incentive compatible (Def~\ref{def:e2e-ic}) if and only if
\begin{itemize}
\item $\M$ has a monotonic allocation rule (Def.~\ref{def:monotone-alloc}).
\item %
The payment rule $p$ follows Myerson's Lemma~\citep{myerson1981optimal}, i.e., $\forall i\in [n], z\in \Z^n, b\in \B$,
\begin{eqnarray}\label{eq:payment}
p_i(b) &=  b_i \cdot \pctr_i(x(b), z)  - \int_{0}^b  \pctr_i(x(y, b_{-i}), z)dy,
\end{eqnarray}
\end{itemize}
\end{proposition}

The proposition follows from the definitions and standard auction theory~\citep{myerson1981optimal}, and we omit the proof. Note that the proposition leaves open the (impractical) possibility of an intricately woven non-monotonic allocation and non-monotonic LLM pair which nevertheless result in a monotonic end-to-end system.

\subsection{Universality of the factorized model}%
We showed above that we obtain an incentive compatible factorized system if both components are monotonic (and the $\pctr$ module is unbiased). One natural question is whether there are other ways to design an incentive compatible auction/LLM system, beyond the factorized model in this paper. 
In this section, we prove a ``revelation principle" type result to show that we can indeed focus on the factorized model without loss of generality.

Consider any incentive compatible \emph{meta LLM-based mechanism} $\M_{\llm} = (x_{\llm}, p_{\llm})$ with an allocation rule $x_{\llm}: \B \times \Z \rightarrow \S$, which maps the bid profile and their original ads contexts to ads summaries directly. In addition, we assume $\M_{\llm}$ is also scale-free:
\begin{assumption}[Scale-freeness]\label{ass:scale-free-llm}
A meta LLM-based mechanism $\M_{\llm}$ is scale-free if $\forall b\in \B, z\in \Z, x_{\llm}(b, z) = x_{\llm}(cb, z)$ for any positive constant $c$.
\end{assumption}
The above scale-freeness property requires $\M_{\llm}$ generates the same ads summaries if all bids are scaled by a constant, which is a natural assumption in practice.%

\begin{theorem}
For any incentive compatible and scale-free meta LLM-based mechanism $\M_{\llm} = (x_{\llm}, p_{\llm})$, there exists an incentive compatible factorized model (having a prominence-based auction with a monotonic allocation function, and a monotonic LLM $\gen$), which achieves the same expected outcome as $\M_{\llm}$.
\end{theorem}

\begin{proof}
Consider any incentive compatible meta mechanism $\M_{\llm} = (x_{\llm}, p_{\llm})$ with an allocation rule $x_{\llm}: \B \times \Z \rightarrow \S$, which maps the bid profile (and their original contents) to ads summaries. Now fix any prominence-based auction $(x, p)$ with allocation $x: \B \times \pctr \rightarrow \Prom$ that is \emph{strictly} monotonic (removing the tie in Definition~\ref{def:monotone-alloc}), and has an inverse $y = x^{-1}$ (up to scaling factors). Note $y$ is a function from a prominence to a bid profile. For example, one can take $x$ to be the proportional allocation rule.
Now we can simply construct an LLM $\gen = x_{\llm} \odot y$, where $\odot$ represents the composition operator of two functions. Given the scale-freeness of $\M_{\llm}$, $\gen$ is well-defined. By construction, the factorized model with the auction $(x, p)$ and the LLM $\gen$ is identical to the given meta mechanism $\M_{\llm}$. The payments follow from Myerson's lemma and are therefore also identical. Note that since the allocation function $x_{\llm}$ of $\M_{\llm}$ is monotonic (since $\M_{\llm}$ is given to be incentive comptibility), and since we chose $x$ to be monotonic, $\gen = x_{\llm} \odot y$ is also monotonic.

\end{proof}

\subsection{Welfare Maximization}

In this section, we consider the auction design problem to maximize the total (expected) welfare of ad summaries generated by LLM, defined as follows,
\begin{eqnarray}
\expect{b \in \F, z}{\sum_{i\in [n]} \expect{s\sim \gen(x(b), z)}{\ctrllm_i(s)}\cdot b_i},
\end{eqnarray}
where $\expect{s\sim \gen(x(b), z)}{\ctrllm_i(s)}\cdot b_i$ captures the expected value achieved by LLM for a fixed bid profile $b$ and the corresponding allocation by the auction. Note, since the auction is incentive compatible, the bid is equal to the private value, therefore, $b_i$ also represents the value of the bidder $i$.
Please note, the total expected welfare doesn't depend on the payment.

\begin{proposition}\label{prop:unbiased-pctr}
Under Assumption~\ref{ass:unbiased-pctr} ($\pctr$ unbiasedness), for any incentive compatible prominence-based auction $\M=(x, p)$ and monotone LLM $\gen$, we have,
\begin{equation}\label{eq:welfare}
\expect{b \in \F, z}{\sum_{i\in [n]} \pctr_i(x(b), z) \cdot b_i} = \expect{b \in \F, z}{\sum_{i\in [n]} \expect{s\sim \gen(x(b), z)}{\ctrllm_i(s)}\cdot b_i}.
\end{equation}
Thus, for a fixed monotone LLM $\gen$, the welfare-maximizing auction mechanism $\M=(x, p)$ will also maximize the expected welfare of final outcome generated LLM.
\end{proposition}

To maximize the total welfare, in general, we need to jointly optimize the LLM $\gen$ and the auction allocation $x$. The LLM can be trained to provide better summaries to maximize clicks while maintaining monotonicity, and the auction can be optimized to provide better prominence allocation. However, the intrinsic combinatorial structure of this problem, e.g., the pCTR model depends on the allocated prominence, makes this optimization problem hard in its general setting.
We instead focus on a simpler setting as a case study in Sec.~\ref{sec:case-study} where we can provide theoretical guarantees on welfare maximization, and also provide empirical evidence for the same in Sec.~\ref{sec:expt}. We note that in this paper, we focus on welfare maximization;
designing revenue-optimal prominence-based auction is an  interesting future direction.%

\section{Case Study: Dynamic Word-length Summary (DWLS) Auctions}\label{sec:case-study}%

The predominant UI for today's online (text-based) advertising (e.g., sponsored search) is to display at most $k$ (e.g. $4$) text ads sequentially in a position-based list~\cite{Varian06} with the ordering determined by an auction using candidates bids and pCTRs. Moreover, each ad can come with a fixed set of optional formats (e.g. merchant rating or phone number) to make the ad more informative, and beyond the ordering, the auction can further optimize total welfare with the available formats subject to total space constraint~\cite{cavallo-richads,aggarwal2022simple}. In this section, we discuss a specific instantiation of our prominence-based auction $+$ LLM framework, and we denote it as the dynamic word-lengh summary (DWLS). 
It is a straightforward and practical extension of today's position-based advertising, 
designed as a minimum example that retains the core elements of the generic framework.

We consider showing up to $k$ text ads subject to a limit on the total number of words (denoted by $L$). In particular, the relative prominence $\prom_i$ for an ad $i$ would be the (intended) fraction of the $L$ words allocated to the ad, i.e. at most $\prom_i\cdot L$ words (up to rounding) for ad $i$. We assume each ad $i$ as an input (i.e. the original full-size text) has $n_i$ words, and if our auction allocates $\prom_i\cdot L < n_i$ words to it, we would use an LLM to compress the original ad into a summary within $\prom_i\cdot L$ words. Note our method, like the position-based UI, compresses ads individually and puts their summaries one-by-one in order.
This is a very viable setup as existing off-the-shelf LLMs can readily perform single ad text summarization with prompting.

Under this setup, the auction's task is to decide both the ordering of the ads to be shown and the fraction of words allocated to each ad. From a technical point of view, if we consider the length of ad summary as an ad format, our setup generalized the existing position-based auction with formats in the sense that we now have access to a set of (continuous) formats dynamically generated by the LLM instead of a fixed given set. The significance of this difference on welfare-maximization is two-fold. Firstly, a much richer set of formats greatly expands the feasible space we can optimize over, and thus the optimal welfare would naturally be greater than the optimal welfare subject to the more restricted format space. Secondly, a more continuous space can also make the mechanism design task easier compared to the combinatorial case, e.g. analogous to solving linear program versus integer program. For example, we show in Section~\ref{sec:gap} under a fairly realistic CTR model, the (exact) welfare maximizing auction becomes very simple, whereas in the combinatorial format case, even approximate optimal auction design can be quite challenging ~\cite{aggarwal2022simple,cavallo-richads}. 

In the rest of this section, we discuss the two modules in more detail. In Section~\ref{sec:prompting}, we talk about the prompting strategy we could use given the allocated prominence (i.e. word limit) to generate good summaries such that the monotonicity property (\ref{def:monotone-llm}) holds. In Section~\ref{sec:gap}, we look at the auction design problem and show the welfare-maximizing auction under a certain CTR model.

\subsection{Prompting Strategy\label{sec:prompting}}

LLMs are very effective in generating summaries that are preferred by humans \citep{liu2022revisiting, zhang2024benchmarking}.
In DWLS we wish to summarize each ad separately. We achieved our objective by few-shot Chain-of-Thought \citep{zhou2022least} prompting. We provide the LLM with very few (3) hand crafted examples of ads, number of words to summarize the ad by, and summary text. We also include key phrases mentioned in the advert as intermediate steps, helping the LLM extract phrases corresponding to the advert that are useful during summarization.

\subsection{Auction Design}
\label{sec:gap}
We consider welfare-maximizing auctions for DWLS, and we work in a simplified case where we are given $n$ eligible ad candidates (with their bids and pCTRs) and also the maximum number of shown ads $k$, so the auction needs to pick exactly $\min(n,k)$ ads to show, and decide the ordering as well as the relative prominence (or equivalently the number of words) of the shown ads. As the total welfare intricately depends on the CTRs of the ads in the way they are shown to the user (i.e. compressed and shown along each other), it's in general implausible to design welfare-maximizing auctions without understanding the CTR. Thus, we start with a fairly realistic CTR model, which is also a straightforward extension of the widely-adopted CTR model in the position-based auction. 

There are several CTR concepts in DWLS: $(1)$ The base CTR when the original ad (i.e. without compression) is shown to the user without other ads. This is given as input to the auction, and in this section we denote it by $\pctr_i$ for an ad $i$. $(2)$ The CTR when the original ad is shown alongside with other ads, and we denote it by $\ppctr_i$. This CTR would naturally decrease for an ad when it's shown after other ads, so it depends on the auction's outcome of ad ordering. $(3)$ The CTR of the ad's summary (i.e. after compression) shown alongside other ads. This further depends on the summarization quality, and we denote it by $\fpctr_i$. Note $\fpctr_i$ would be the CTR in the calculation of welfare for a shown ad $i$, i.e. the $\pctr_i(x(b),z)$ term in Equation~\eqref{eq:welfare}. These variants are analogous to the position-based case where there are known as position-$1$ pCTR, position-normalized pCTR and position-normalized-formatted pCTR respectively. 

Similar to CTR models in practice for the position-based case, we consider a factorized model where $\ppctr_i=\pctr_i\cdot\posnorm_i$ with $\posnorm_i$ being a discount factor for the position (i.e. order) of ad $i$ in the list of shown ads. The position discount factor only depends on the position and is provided as input to the auction as $k$ fixed numbers (e.g. $1.0, 0.9, 0.81, \ldots$) for the positions respectively. The factors are in $[0,1]$ and decreases for later positions. Moreover, $\fpctr_i$ would be $\ppctr_i$ further multiplied by a compression discount factor capturing the summarization quality. Since we focus on auction design, our CTR model in this section doesn't explicitly consider the LLM quality, so we assume the LLM can faithfully follow the instructed word limit and do equally well at summarizing any given ad conditioned on the number of allocated words, and thus the compression discount factor is (only) a function on the number of words. Technically this means $\fpctr_i=\ppctr_i\cdot f(\prom_i)$ for some fixed function $f(\cdot)$ since $\prom_i$ is (up to scaling by $L$) equivalent to number of words allocated to ad $i$. In the rest of the section, we focus on the family of $f(\prom_i)=\prom_i^\beta$ for $\beta\in(0,1]$ and characterize the welfare-maximizing auction in this setting. Note as $\prom_i \in [0,1]$, such family of polynomial functions naturally serve as a compression discount factor in $[0,1]$. For simplicity, it's helpful to think of our CTR model in a setting where the full length of all ads are roughly the same, and are all much longer than the number of available words to show ad summaries. 

Formally speaking, given a set of $n$ ad candidates, we denote $\vbid$ as the vector of bids with $\bid_i$ being the bid of ad $i$, $\vpctr$ as the vector of input (base) CTRs, and $r_1\geq\ldots\geq r_k$ as the position discount factor of the $k$ ad positions respectively. In DWLS, we write the auction's allocation explicitly as two vectors $\vposnorm$ and $\prom$, where $\vposnorm$ captures which (up to $k$) ads to show and their position (i.e. ordering) and $\prom$ is the vector of relative prominence, which lies in the $n$-dimensional simplex. The $i$-th entry of $\vposnorm$ would be $r_t$ if ad $i$ is picked by the auction to show at position $t\in[1,\ldots,k]$ or $0$ if not picked. The total welfare is
\[
(\vbid\odot\vpctr\odot\vposnorm)\cdot f(\prom)
\]
where $\odot$ denotes entry-wise multiplication, $\cdot$ is dot product, and $f(\prom)$ is the vector of applying $f(\cdot)$ to each entry of $\prom$.

It becomes clear from the welfare formulation that the factorized CTR model allows the auction to pick an optimal $\prom$ conditioned on $\vposnorm$ (i.e. the shown ads and their positions), and we first characterize the optimal $\prom$ allocation rule.  
\begin{definition}[Generalized Proportional Allocation]
\label{def:gpa}
Given $n$ ads with $\vbid, \vpctr$, $\vposnorm$, and denote $\vecpm$ as the vector with $\ecpm_i = \bid_i\cdot\pctr_i\cdot\posnorm_i$, the generalized proportional allocation with parameter $\alpha$ allocates as follows
\begin{eqnarray}
\prom_i = \frac{(\ecpm_i)^\alpha}{\sum_{j=1}^n (\ecpm_j)^\alpha},\quad \forall i=1,\ldots, n  
\label{eq:gpa}
\end{eqnarray}
Note in DWLS there will be exactly $\min(n,k)$ non-zero entries in $\vecpm$.
\end{definition}
We will show that the welfare-maximizing auction corresponding to any CTR model we consider with $f(\prom_i)=\prom_i^\beta$ for some $\beta \in (0,1]$ would use a generalized allocation rule with appropriately chosen $\alpha$. As a special case, consider when $\beta=1$ so $f(\prom_i)=\prom_i$ and we want to maximize $\vecpm\cdot\prom$. Since $\prom$ lies in the $n$-dimensional simplex, the optimal $\prom$ is to put weight only on the ad(s) in the set of $\argmax_{i\in [n]}\ecpm_i$, and this correspond to the generalized proportional allocation with $\alpha=\infty$. This is a fairly trivial case, and also a less realistic example in our family of CTR models. In particular, the $f(\cdot)$ function in a more realistic CTR model would capture a \emph{diminishing return} phenomenon w.r.t. prominence (or equivalently the number of words), i.e., the marginal improvement of CTR vanishes as we add more and more words to the summary. For example, the CTR increase from a $10$-word ad summary to a $20$-word summary is typically much higher than the CTR increase from $40$ to $50$ words. Technically, a strictly concave $f(\cdot)$ would capture this effect, which corresponds to $\beta\in (0,1)$ with stronger diminishing return effect for smaller $\beta$. 
\begin{theorem}\label{thm:optimality-proportional}
When $f(\prom_i)=\prom_i^\beta$ for $\beta\in (0,1)$ in our factorized CTR model, the optimal $\prom$ conditioned on $\vposnorm$ follows the generalized proportional allocation with $\alpha=1/(1-\beta)$, and the optimal welfare conditioned on $\vposnorm$ is $\|\vecpm\|_\alpha$, i.e. the $\alpha$-norm of the vector $\vecpm$ in Definition~\ref{def:gpa}.
\end{theorem}
\begin{proof}
Consider the vector $\vecpm$ with $\ecpm_i = \bid_i\cdot\pctr_i\cdot\posnorm_i$, and denote $\vec{x}$ as the vector $f(\prom)$ (i.e. $x_i=\prom_i^\beta$). The constraint that $\prom$ lies in the simplex is equivalent to the $p$-norm of $x$ being $1$ for $p=1/\beta$, and the optimization of $\prom$ is equivalent to maximize $\vec{x}\cdot\vecpm$ subject to $\|x\|_p=1$. By H\"older's inequality, we know $\vec{x}\cdot\vecpm\leq \|x\|_p\|\vecpm\|_q = \|\vecpm\|_q$ where $q$ is the dual-norm, i.e. $q=1/(1-1/p)=1/(1-\beta)$. Since $\|\vecpm\|_q$ is fixed, the optimal welfare is achieved with the $\vec{x}$ making the inequality an equality. Since $p,q\in (1,\infty)$, we can explicitly write the $\vec{x}$ that gives the equality case in H\"older's inequality, which is $x_i^p=(\ecpm_i)^q/\sum_{j\in[n]}(\ecpm_j)^q$. Since $\prom_i=x_i^p$, we get the generalized proportional allocation with $\alpha=q=1/(1-\beta)$.
\end{proof}
With the above theorem, choosing the optimal $\vposnorm$ in the welfare-maximizing auction becomes straightforward. As the optimal welfare conditioned on $\vposnorm$ is exactly $\|\vecpm\|_\alpha$, we should pick the ads and their ordering to maximize $\|\vecpm\|_\alpha$. For any $\alpha\in (1,\infty)$ and $\vecpm=\bid\odot\vpctr\odot\vposnorm$, the optimal $\vposnorm$ is clearly incurred by picking the top $\min(n,k)$ ads with the highest product of bid and pCTR and show them in the same order according to that product. This completes the optimal allocation of the welfare-maximizing auction. It is straightforward to check this allocation rule is monotone since both $\posnorm_i$ and $\prom_i$ are monotonic in $\bid_i$. Since we consider IC mechanisms, the pricing rule follows from Myerson's Lemma, which is also straightforward to calculate in our factorized CTR model.

We note that it may not be immediately clear from the way we present DWLS that it's a strict instantiation of our general prominence-based auction $+$ LLM model. To cast it in our model, note the auction only need to output the relative prominence $\prom$ as in the general model, and the set of shown ads as well as their ordering can be inferred from $\prom$. That is, the LLM in our general model would be a meta-LLM in DWLS that takes in the $\prom$ vector, finds the (up to $k$) ads with non-zero weights in $\prom$, then calls a real LLM to compress the ads one-by-one separately with the respective word limits according to $\prom$, and arrange the individual ad summaries in the same order as their weights in $\prom$.

\section{Experimental Results}\label{sec:expt}
In this section, we present an empirical study on a synthetic data set for the DWLS setting. Through these experiments we first verify the validity of our proposed framework to generate ad summaries with the auction and LLM framework. %
We then demonstrate the efficiency of our mechanism, i.e., our generalized proportional allocation rule and prompting strategy to the LLM can provide more efficient allocation compared with some simple benchmarks. The LLMs we used to generate synthetic ads data and perform summarization task are both Gemini 1.0 Pro.

\begin{table*}[t]
\centering
\begin{tabular}{@{}p{2.8in}cccc@{}}
\toprule
Generated Ad & & Generated Summary A & & Generated Summary B \\ \midrule 
 {\footnotesize $\langle url1\rangle$ Feeling intimidated by golf? Our beginner-friendly clinics offer a relaxed environment to learn fundamentals and etiquette. Get comfortable with the swing, chipping, and putting in a supportive group setting. No experience needed, just bring your enthusiasm!} & & \multirow{3}{*}{\parbox{3.6cm}{\footnotesize {\textcolor{blue}{{$\langle url1\rangle$}}} Golf clinics for beginners! Learn basics, etiquette, \& practice swing, chipping, putting in a relaxed group setting. No experience required, just bring your enthusiasm!\\ {\textcolor{blue}{{$\langle url2\rangle$}}} Master golf at home! PGA Pro video lessons offer swing analysis, drills \& coaching. Learn on your own schedule.\\ {\textcolor{blue}{{$\langle url3\rangle$}}} Fun, safe learning! Build skills, sportsmanship, \& friendships. Expert coaching ignites love for the game.}} && \multirow{3}{*}{\parbox{3.6cm}{ \footnotesize {\textcolor{blue}{{$\langle url1\rangle$}}} Conquer your golf fears! Our beginner clinics in a relaxed setting teach you the basics and proper etiquette. Learn swing, chipping, and putting with supportive peers. No experience required, just bring your enthusiasm!\\ {\textcolor{blue}{{$\langle url2\rangle$}}} Level up your golf game from home! Get personalized video lessons with PGA pros. Master swing, drills, and course strategy at your own pace.}}\\
\cmidrule{1-1}
 {\footnotesize $\langle url2\rangle$ Busy schedule? Learn from the comfort of your couch! Connect with experienced PGA pros for tailored video lessons. Get swing analysis, drills, and course strategy guidance at your own pace. Achieve your golfing goals on your own terms.} & &   \\
 \cmidrule{1-1}
 {\footnotesize $\langle url3\rangle$ Spark a lifelong love for the game! Our programs provide a fun and safe space to learn golf basics, develop sportsmanship, and make new friends. Qualified instructors guide young golfers through proper techniques and etiquette, fostering a passion for the green.} & &  \\
\bottomrule
\end{tabular}
\caption{{Synthetic ads example generated by LLMs for the query ``Learning golf'' and the output summarized by LLM in accordance with the prominence provided by GPA. We used $\alpha=2$ and total word length$=60$. For summary A, the input ecpms to the auction (bid times base pctr) were [0.645, 0.641, 0.617], the relative prominence output by the auction was [0.417, 0.333, 0.250] implying a target word length distribution of [25, 20, 15]. For summary B, ecpm=[0.764, 0.710, 0.113], the relative prominence=[0.583, 0.409, 0.007] and the target word lengths=[35, 25, 0].}}
\label{tab:synthetic_ad_example}
\end{table*}

\subsection{Data Generation}\label{sec:data-generation}
Using a large LLM (Gemini 1.0 Pro), we produce a set of \emph{synthetic} and \emph{anonymous} advertisements. For each generated query, the LLM generates between two and four ads. For example, we show an example contains 3 synthetic ads generated by a LLM for the query ``\emph{Learning golf}" in Table~\ref{tab:synthetic_ad_example}.
Following the standard assumption in ad auctions (e.g., \cite{ostrovsky2011reserve}), we assume the bid $b_i$ of each ad $i$ follow a log-normal distribution $\textsc{LogNormal}(0.5, 1)$. In addition, the click-through rate $CTR_i$ of each ad $i$ is independently and identically (i.i.d) sampled from a uniform distribution $\textsc{Unif}[0, 1]$. %
Note the $CTR_i$ is the \emph{base} click-through rate of ad $i$ if it is shown to the user solely.

\subsection{Evaluation Model}
To evaluate the real clicks (and thus welfare), we need user feedback from a live production system. This requires deploying our proposed framework in real online advertising production to get human evaluation, which is beyond the scope of this work. For the purpose of evaluation in this paper, we propose a synthetic evaluation model. In particular, we first define the CTR function used to evaluate the performance of different mechanisms,
\begin{definition}[Click-through Rate function for evaluating welfare]
For any ad summaries $s\in \S$ generated by LLM and the associated original contexts $z$, the CTR of each ad $i\in [n]$ is taken to be
\begin{eqnarray}
CTR_i(s, z) = CTR_i \cdot f_i(s, z) \cdot norm_i(s),
\end{eqnarray}
where $CTR_i$ is the base click-through rate of ad $i$, $f_i(s, z)$ models the CTR multiplier due to the summary quality of each ad $i$ as a function of the summaries $s$ and the original contexts $z$, and $norm_i(s)$ quantifies the UI normalizer for each ad $i$ in the summary. %
\end{definition}
In the DWLS setting, $norm_i(\cdot)$ is simply a position discount factor, i.e., $norm_i(\cdot)$ only depends on the rank of ad $i$'s summary in $s$ and this UI normalizer will be lower when the its rank is lower. We use functionals of the \textsc{ROUGE} metric~\citep{Lin04Rouge} to quantify summary quality $f_i(s,z)$. At a high level, the (variant of) ROUGE metric we use is a score between $0$ and $1$, and captures both the length of the summary $s$ and the relevance between $s$ and $z$. As the LLM mostly can summarize very well, the main factor for the ROUGE score in our case becomes the length we allocate to $s$, and thus the ROUGE and relative prominence roughly follow a linear relationship. Consequently, the evaluation model aligns qualitatively with our theoretical model in Section~\ref{sec:case-study} when we take the $f_i(s,z)$ to be some concave polynomial function of ROUGE.

\subsection{Benchmarks}
We compare the proposed generalized proportional auction mechanism against two natural baseline auction mechanisms, one that doesn't utilize the power of LLMs for summarization, and the other which doesn't involve optimizing the auction design. 
\begin{itemize}[leftmargin=*]%
\item \emph{Greedy Auction}: Ads are shown in a list respecting the $CTR_i \cdot b_i$ ranking. The original creatives of ads are shown (i.e. with no summarization) until we run out of space. If an ad doesn't fit in the remaining space then it cannot be shown. This baseline evaluates the performance of an auction mechanism that doesn't utilize the power of LLMs to summarize.
\item \emph{Position Auction with Fixed Length}: Ads are shown in a summary paragraph which gives equal number of words for each ad. 
This baseline evaluates the performance of LLMs to summarize, but without optimizing the auction design. 
\end{itemize}
We compare these baselines against generalized proportional auction with an appropriate $\alpha$ (defined below) combined with the LLM output ("GPA+LLM") by comparing the value of average welfare realized by each of these approaches.

\subsection{Set-up}
First, we generate 1000 random queries. Then, we use the data generation method from Section~\ref{sec:data-generation} to generate synthetic and anonymous ad creatives for each query.
Throughout the experimental section, we set the $norm_i(s) = 0.9^{(rank_i(s) - 1)}$, where $rank_i(s)$ is the rank of ad $i$'s summary in $s$. We vary function $f_i$ in our experimental results. Specifically, we set $f_i(s,z) = (\ROUGE(s_i, z_i))^{\beta}$ where we set $\beta$ to $1/2,\ 1/3\ \&\ 1/4$. Motivated by Theorem~\ref{thm:optimality-proportional} we choose $\alpha= \frac{1}{1-\beta}$ in our generalized proportional auction. We note in practice when $f_i(s,z)$ is not known, it is conceivable that some polynomial function with $\beta\in (0,1)$ qualitatively approximates it well. Thus generalized proportional auction with the corresponding (but unknown) $\alpha$ still has good performance, and it is very practical to tune for a good $\alpha$.

\subsection{Results}

\subsubsection{Qualitative} In table~\ref{tab:synthetic_ad_example}, we present a qualitative result through a sample (end to end) instantiation of our framework for two different sets of bids. In each setting, the bids and pctrs are converted by the auction to prominence values (number of words), and the LLM faithfully generates a summary based on the prominence. We note how the first ad gets a significantly larger fraction of the space in the second (skewed) setting of bids.

\begin{figure}%
\centering
\includegraphics[width=0.49\textwidth]{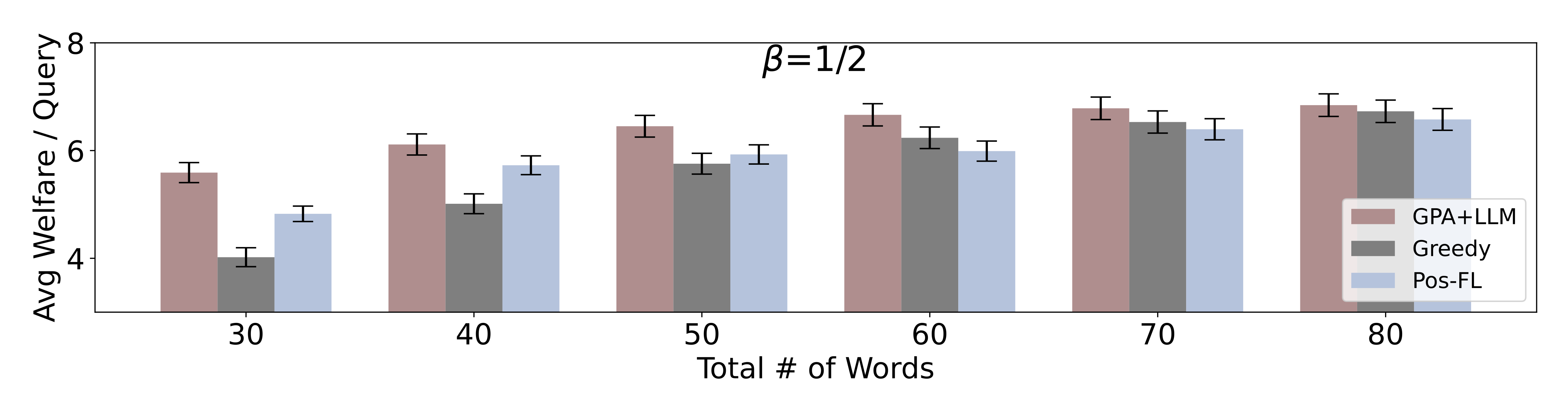}
\includegraphics[width=0.49\textwidth]{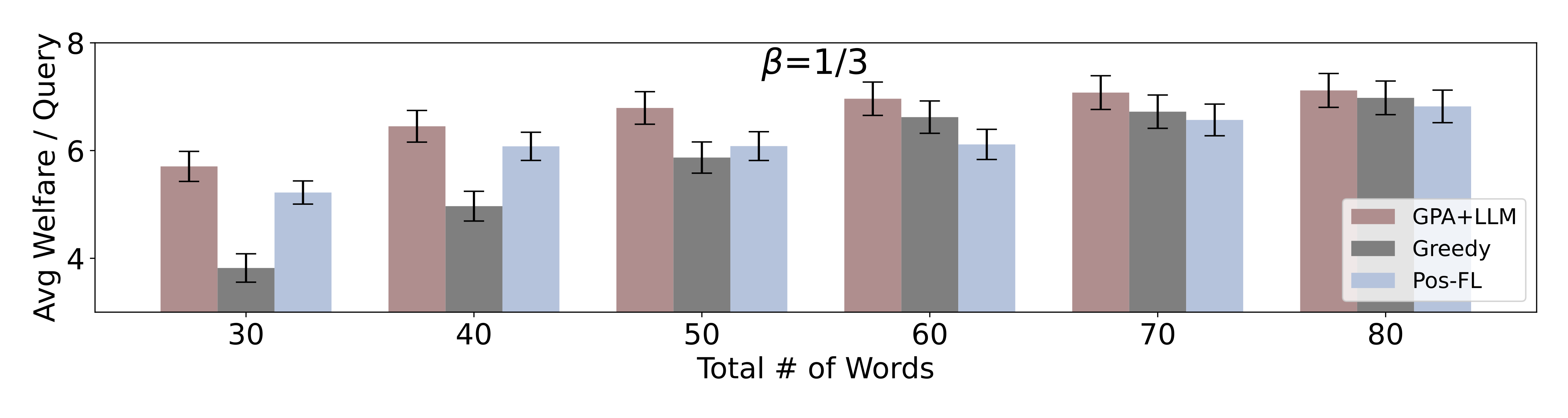}
\includegraphics[width=0.49\textwidth]{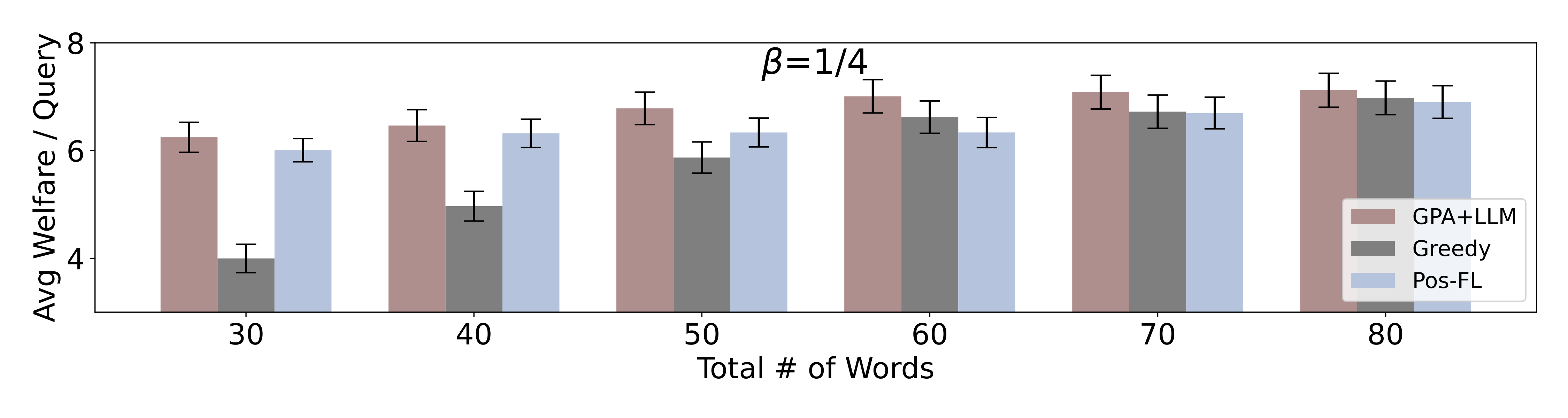}
\caption{Total welfare with different choices of total number of words for summarization.}
\label{fig:welfare_vs_total_words}
\end{figure}

\subsubsection{Efficiency}
In figure~\ref{fig:welfare_vs_total_words}, we compare the two baselines, namely the position auction with fixed length ("POS-FL") and the greedy auction against GPA+LLM with $\alpha=\tfrac{1}{1-\beta}$. Note that the behavior of the two baselines, Greedy and Position Auction, do not depend on the value of $\beta$. Furthermore, the reward for Greedy also does not depend on $\beta$ because if it shows an ad then it shows the entire ad text yielding a ROUGE score of 1. We observe from  figure~\ref{fig:welfare_vs_total_words} that GPA+LLM does better than both the baselines showcasing the value of utilizing the power of LLM along with optimizing the auction allocation. As the total number of words increases, the baselines catch up to the proposed GPA+LLM approach since the LLM is largely not required for summarizing the creatives, as all the original ad creatives can fit in many cases.
Comparing across the two baselines, when the total number of words is small, then greedy often leaves space unused since it can only show an ad creative in its entirety (since it does not use LLM to resize ads) and hence performs worse than position auction which does resize ads. When the total number of words is larger, then greedy starts doing better than position auction, because it is often better to split the space between fewer ads (e.g., two ads when the third ad's bid is low).

\section{Conclusions and Future Directions}\label{sec:conclusion}
This paper develops a factorization that enables auction based allocation for general LLM-based summarizations. 
The paper studies an instantiation that shows the near-optimality of the generalized proportional auction under a certain class of parameterized click through rate models for the LLM-generated summaries.
While the empirical section in the paper works with a {\em fixed} LLM and the dynamic word-length summarization interface, it is worth looking into the use of finetuning to work in conjunction with appropriate auction design to adapt to user behavior on other classes of flexible user interfaces.

\bibliographystyle{plainnat}
\bibliography{ref}

\appendix

\end{document}